\documentclass[letterpaper, 10pt]{IEEEtran}
\IEEEoverridecommandlockouts 
\usepackage{amsmath}
\usepackage{amsthm}
\usepackage{amsfonts}
\usepackage{amssymb}
\usepackage{hyperref}
\usepackage{algorithm, algpseudocode}  
\usepackage{array}                      
\usepackage{cite}
\usepackage{wrapfig}
\usepackage{lipsum}
\usepackage{graphicx}
\usepackage{mathtools}
\usepackage{comment}
\usepackage{xcolor}
\usepackage{mathrsfs}

\usepackage[compact]{titlesec}

\usepackage{blindtext}
\usepackage{tabularx,booktabs}
\newcolumntype{C}{>{\centering\arraybackslash}X} 

\usepackage [english]{babel}
\usepackage [autostyle, english = american]{csquotes}
\MakeOuterQuote{"}

\theoremstyle{definition}
\newtheorem{theorem}{Theorem}  
\newtheorem{assumption}{Assumption}

\newtheorem{remark}{Remark}

\usepackage{enumitem}

\graphicspath{{Figures/}}

\newcommand{\update}[1]{\textcolor{black}{#1}}
\newcommand{\alert}[1]{\textcolor{black}{#1}}

\begin{document}
\pagenumbering{gobble}

\title{\update{Distributed State Estimation with Deep Neural Networks for Uncertain Nonlinear Systems under Event-Triggered Communication}}
	
\author{Federico M. Zegers, Runhan Sun, Girish Chowdhary, and Warren E. Dixon
\thanks{Federico M. Zegers, Runhan Sun, and Warren E. Dixon are with the Department of Mechanical
and Aerospace Engineering, University of Florida, Gainesville, Florida 32603, e-mail: fredzeg@ufl.edu, runhansun@ufl.edu, and wdixon@ufl.edu.}  
\thanks{Girish Chowdhary is with the Department of Agricultural and Biological Engineering, University of Illinois, Urbana-Champaign, IL 61801, e-mail: girishc@illinois.edu}
\thanks{This research is supported in part by A Task Order contract with the Air Force Research Laboratory, Munitions Directorate at Eglin AFB, NSF award 1509516, Office of Naval Research Grant N00014-13-1-0151, NEEC award number N00174-18-1-0003, AFOSR award numbers FA9550-18-1-0109 and FA9550-19-1-0169. Any opinions, findings and conclusions or recommendations expressed in this material are those of the author(s) and do not necessarily reflect the views of sponsoring agencies.}
\thanks{Manuscript submitted to \emph{IEEE Transactions on Automatic} Control}
}
	
\maketitle
		
\begin{abstract} 
Distributed state estimation is examined for a sensor network tasked with \update{reconstructing} a system's state through the use of a distributed and event-triggered observer. Each agent in the sensor network employs a deep neural network (DNN) to approximate the uncertain nonlinear dynamics of the system, which is trained using a multiple timescale approach. Specifically, the outer weights of each DNN are updated online using a Lyapunov-based gradient descent update law, while the inner weights and biases are trained offline using a supervised learning method and collected input-output data. The observer utilizes event-triggered communication to promote the efficient use of network resources. A nonsmooth Lyapunov analysis shows the distributed event-triggered observer \update{has a} uniformly ultimately bounded state reconstruction \update{error}. A simulation study is provided to validate the result and demonstrate the performance improvements afforded by the DNNs.

\end{abstract} 

\section{Introduction}
A \update{wireless sensor network (WSN)} is defined as a multi-agent system composed of autonomous sensors scattered over an area to monitor desired phenomena and connected through wireless communication links \cite{mostafaei2018energy}. By sharing partially observable state measurements of a system with their neighbors and leveraging a consensus algorithm, WSNs are capable of estimating the state of a system in a distributed fashion \cite{ge2019distributed}. This technique is called distributed state estimation, and it allows each sensor in the WSN to reconstruct the entire system state through local and cooperative information sharing despite each agent only being able to measure part of the system's state. Distributed state estimation does not require a data fusion center; therefore, it is a preferable state estimation strategy since it can better accommodate each agent's limited computing capacity, eliminate single points of failure, and promote scalability. 

In \cite{he2018consistent}, the authors developed a decentralized consensus-based observer capable of performing stable distributed state estimation using adaptive weights that were generated by solving a semi-definite program. Given that agents within a WSN may be powered by limited portable energy sources, results like \cite{battistelli2018distributed} and \cite{ge2018threshold} developed distributed state estimation strategies with event-triggered communication as a means to conserve energy and network resources. Similarly, the authors in \cite{yu2020event} developed a distributed state observer with stochastic event-triggered communication for a linear time-varying system, which improves upon \cite{he2018consistent,battistelli2018distributed,ge2018threshold}.

While results such as \cite{he2018consistent, battistelli2018distributed, ge2018threshold, yu2020event} provide valuable contributions towards the literature on distributed state estimation, these results, like many of the techniques in \cite{ge2019distributed}, are focused on known linear systems. Results on distributed state estimation for systems with uncertain and/or nonlinear dynamics are scarce but well motivated. Additionally, the improved computing power of modern processors along with data availability encourages the development of a distributed observer capable of employing machine learning techniques as a means to improve state reconstruction. However, the update laws used to train the weights and biases of deep neural networks (DNNs) do not typically have a stability analysis, which has mitigated their use for online estimation and control. Conversely, in this paper, we develop a Lyapunov-based update law for the outer weights of a DNN and prove the stability of the DNN-based function approximator for an uncertain nonlinear system.

Recently, the authors in~\cite{joshi2020design,joshi2019deep,Joshi.Virdi.ea.2020} developed a model reference adaptive control architecture that utilizes a DNN as the adaptive element while ensuring that the estimation error is uniformly ultimately bounded (UUB) via a Lyapunov-based stability analysis. These works are among the first to employ DNNs for real-time control while providing a formal stability assurance. The key innovation lies in the update of the DNN weights. The outer layer weights evolve according to a real-time analysis-based update law that ensures stability, and the inner layer weights are modified using batch updates. Based on this observation, the authors in \cite{Sun.Greene.eatoappear} developed a DNN adaptive controller for an uncertain nonlinear dynamical system capable of asymptotically tracking a desired trajectory. In this result, the outer layer weights of the DNN are updated in real-time using a Lyapunov-based update law, while the inner layer weights are updated using a data-driven supervised learning algorithm, i.e., the Levenberg-Marquardt algorithm. \update{The results in \cite{joshi2020design} and \cite{Sun.Greene.eatoappear} show that multiple timescale learning with DNNs can yield improved performance when compared to traditional adaptive techniques.} 

Results such as \cite{parisini1995receding,zhang2020near,karg2020efficient} provide alternative methods of employing DNNs for control and estimation. The work in \cite{parisini1995receding} developed a stabilizing regulator for a known discrete-time nonlinear system using a receding-horizon optimal control scheme. The proposed feedback control law is computed offline, where a DNN is used to approximate the receding-horizon regulator. In \cite{zhang2020near}, a model predictive control (MPC) policy for a known linear parameter-varying system is approximated using a DNN that is trained online with supervised learning. The authors in \cite{karg2020efficient} provide a DNN-based strategy to approximate an MPC control policy for a known linear system that reduces memory requirements when compared to other approximation techniques. While these methods provide computationally efficient control strategies, they are not readily applicable to uncertain nonlinear systems.

Inspired by \cite{yu2020event,joshi2020design}, and~\cite{Sun.Greene.eatoappear}, we develop an adaptive event-triggered distributed state observer that utilizes DNNs as a means to improve state reconstruction for an uncertain nonlinear system. Using a nonsmooth Lyapunov stability analysis, we prove that our observer is capable of UUB state reconstruction while being robust to a bounded exogenous disturbance. Similar to \cite{joshi2020design} and \cite{Sun.Greene.eatoappear}, we develop a multiple timescale learning strategy. In particular, the outer layer weights of each DNN are adjusted online using a Lyapunov-based update law that uses real-time feedback to ensure stability, while the inner layer weights and biases are updated offline using a supervised learning algorithm with collected input-output data. The theoretical findings of our work are validated through a simulation study. The \update{observer is} capable of reducing the root-mean-square state estimation error of each agent in the WSN by approximately $60\%$ when compared to an identical simulation, where the DNN inner weights and biases are held constant. 

The rest of the paper is organized as follows. Section \ref{Prelim. Section} introduces notation and necessary concepts about graphs. Section \ref{Section Dynamics and Network} introduces the system model and the sensor network. Section \ref{Objection Section} precisely formulates the goal of this work. Section \ref{Observer Section} develops the proposed observer and closed-loop dynamics of the error system encoding the state reconstruction error. Sections \ref{Stability Analysis} and \ref{Simulation Section} prove the UUB state reconstruction of the proposed observer and investigate the performance of the development, respectively. Section \ref{Concluding Section} summarizes this work and suggests possible future directions.

\section{Preliminaries} \label{Prelim. Section}
\subsection{Notation}
Let $\mathbb{R}$ and $\mathbb{Z}$ denote the set of real numbers and integers, respectively. We also write $\mathbb{R}_{\geq x}\triangleq\left[x,\infty\right),$
$\mathbb{R}_{>x}\triangleq\left(x,\infty\right),$ $\mathbb{R}_{<x}\triangleq\left(-\infty,x\right),$
$\mathbb{Z}_{\geq x}\triangleq\mathbb{R}_{\geq x}\cap\mathbb{Z},$
and $\mathbb{Z}_{>x}\triangleq\mathbb{R}_{>x}\cap\mathbb{Z}$ for
$x\in\mathbb{R}.$ For $p,q,n,u,v\in\mathbb{Z}_{>0}$, the $p\times q$
zero matrix and the $p\times 1$ zero column vector are denoted by
$0_{p\times q}$ and $0_{p},$ respectively. 
The $p\times p$ identity
matrix and the $p\times 1$ column vector of ones are denoted by $I_{p}$
and $1_{p},$ respectively. 
The Euclidean norm of a vector $r\in\mathbb{R}^{p}$
is denoted by $\left\Vert r\right\Vert \triangleq\sqrt{r^{\top}r}.$ Given a positive integer $M$, let $[M]\triangleq\left\{1,2,...,M\right\}$. The Kronecker product of $A\in\mathbb{R}^{p\times q}$ and $B\in\mathbb{R}^{u\times v}$
is denoted by $\left(A\otimes B\right)\in\mathbb{R}^{pu\times qv}$. The block diagonal matrix whose diagonal blocks consist of $G_{1},G_{2},...,G_{n}$ is denoted by $\text{diag}\left(G_{1},G_{2},...,G_{n}\right)$. The $i^{\text{th}}$, maximum, and minimum eigenvalues of a symmetric matrix $G\in\mathbb{R}^{p\times p}$ are denoted by $\lambda_{i}(G)\in\mathbb{R}$, $\lambda_{\max}(G)\in\mathbb{R}$, and $\lambda_{\min}(G)\in\mathbb{R}$, respectively. The trace of a square matrix $A\in\mathbb{R}^{p\times p}$ is denoted by $\text{tr}(A)$. Let $\text{vec}(\cdot)$ denote the vectorization transformation that converts a matrix into a column vector. The symbol $\mathcal{L}_{\infty}$ denotes the set of essentially bounded measurable functions, i.e., given the Lebesgue measurable function $f:\mathbb{R}\rightarrow\mathbb{R}$, $f\in\mathcal{L}_{\infty}$ if and only if $\inf\left\{ C\geq0:\left|f\left(x\right)\right|\leq C\text{ for almost every \ensuremath{x}\ensuremath{\in}\ensuremath{\mathbb{R}}}\right\} \in\mathbb{R}_{\geq0}$. The symbol $\circ$ denotes function composition, i.e., given suitable functions $f$ and $g$, $(f\circ g)(x)=f(g(x))$.

\subsection{Graphs}
Let $\mathcal{G}\triangleq(\mathcal{V},\mathcal{E},\mathcal{A})$ denote a static, weighted, and undirected graph with node set $\mathcal{V}\triangleq[N]$, for some $N\in\mathbb{Z}_{>0}$, edge set $\mathcal{E}\subseteq\mathcal{V}\times\mathcal{V}$, and symmetric weighted adjacency matrix $\mathcal{A}\triangleq [a_{ij}]\in\mathbb{R}^{N\times N}$. 
The edge $(i,j)\in\mathcal{E}$ if and only if node $i$ can send information to node $j$. 
Since $\mathcal{G}$ is undirected, $(i,j)\in\mathcal{E}$ if and only if $(j,i)\in\mathcal{E}$.  
An undirected graph is connected if and only if there exists a sequence of edges in $\mathcal{E}$ between any two distinct nodes. The neighbor set of node $i$ is denoted by $\mathcal{N}_i\triangleq\{j\in\mathcal{V}:(j,i)\in\mathcal{E}\}$. Within this work, no self-loops are considered, and therefore, $a_{ii}\triangleq0$ for all $i\in\mathcal{V}$. Moreover, $a_{ij}>0$ if $(j,i)\in\mathcal{E}$, and $a_{ij}=0$ if $(j,i)\notin\mathcal{E}$. The degree matrix of $\mathcal{G}$ is defined as a diagonal matrix such that $\Delta\triangleq[\Delta_{ij}]\in\mathbb{R}^{N\times N}$, where
$\Delta_{ij}\triangleq0$ for all $i\neq j$, and $\Delta_{ii}\triangleq\sum_{j\in\mathcal{V}} a_{ij}.$ The Laplacian matrix of $\mathcal{G}$ is denoted by $L\in\mathbb{R}^{N\times N}$ and defined as $L\triangleq\Delta-\mathcal{A}.$

\section{System Dynamics and Network Topology} \label{Section Dynamics and Network}
Consider a system whose uncertain model is given by
\begin{equation}
    \dot{x}_{0}(t)=f(x_{0}(t))+d(t),
\label{plant dynamics}
\end{equation}
where $x_0:[0,\infty)\to\mathbb{R}^{n}$ denotes the state, $f:\mathbb{R}^{n}\to\mathbb{R}^{n}$ denotes the uncertain nonlinear dynamics, and $d:[0,\infty)\to\mathbb{R}^{n}$ denotes an exogenous disturbance. Furthermore, consider a sensor network composed of $N$ agents, which are indexed by $\mathcal{V}$. For each $i\in\mathcal{V}$, agent $i$ is capable of continuously measuring the output $y_i:[0,\infty)\to\mathbb{R}^{m}$, where the output measurement is given by
\begin{equation}
    y_i(t)=C_{i}x_{0}(t)
\label{y_i}
\end{equation}
such that $C_i\in\mathbb{R}^{m\times n}$ denotes the known output matrix of agent $i$. The agents in the sensor network may have different sensing capabilities, where each agent may be able to measure a different component of the system's state. Each agent is also capable of intermittently communicating with its neighbors, where the flow of information between the agents in the sensor network is modeled through the static and undirected communication graph $\mathcal{G}=(\mathcal{V},\mathcal{E},\mathcal{A})$.\footnote{Future works can consider measurement and process perturbations as well as networked communication constraints such as delayed information and packet dropouts. However, we work in the nominal setting for simplicity and to highlight the novelty of the distributed state estimation strategy.} The following assumptions are used in the development of the result.
\begin{assumption}
The function $f$ is locally Lipschitz. 
\label{Assumption 1}
\end{assumption}
\begin{assumption}
The disturbance is continuous and bounded, i.e., there exists a $d_{\max}\in\mathbb{R}_{>0}$ such that $\Vert d(t) \Vert \leq d_{\max}$ for all $t\geq 0$.
\label{Assumption 2}
\end{assumption}
\begin{assumption}
The communication graph $\mathcal{G}$ is connected for all $t\geq 0$.
\label{Assumption 3}
\end{assumption}

\section{Objective} \label{Objection Section}
The objective is to develop a distributed observer capable of reconstructing the state of an uncertain nonlinear dynamical system. The distributed observer must also be event-triggered to promote the efficient use of network resources. Furthermore, we wish to design an observer that concurrently utilizes online and offline learning strategies to ensure stability and enable improved state reconstruction. To quantify the objective, let the state estimation error $e_{1,i}:[0,\infty)\to\mathbb{R}^{n}$ of agent $i\in\mathcal{V}$ be defined as
\begin{equation}
    e_{1,i}(t)\triangleq \hat{x}_{i}(t)-x_{0}(t),
\label{e1i}
\end{equation}
where $\hat{x}_i:[0,\infty)\to\mathbb{R}^{n}$ denotes the estimate of $x_{0}(t)$ as computed by agent $i$. The state estimation error $e_{1,i}(t)$ is an unmeasurable signal that is used only in the analysis. To facilitate the use of event-triggered control, let $\tilde{x}_i:[0,\infty)\to\mathbb{R}^{n}$ denote agent $i$'s sampled state estimate. For example, if agent $i$ uses a zero-order hold policy and samples its state estimate according to the increasing sequence $\{T_{k}\}_{k=0}^\infty$, then the sampled state estimate is given by $\tilde{x}_{i}(t)=\hat{x}_{i}(T_{k})$ for all $t\in[T_{k},T_{k+1})$ and each $k\in\mathbb{Z}_{\geq 0}$. The sampled state estimation error $e_{2,i}:[0,\infty)\to\mathbb{R}^{n}$ is defined as
\begin{equation}
    e_{2,i}(t)\triangleq \tilde{x}_{i}(t)-\hat{x}_{i}(t).
\label{e2i}
\end{equation}
The estimated output of the system with respect to agent $i$ is denoted by $\hat{y}_i:[0,\infty)\to\mathbb{R}^{m}$, where the output estimation error $e_{3,i}:[0,\infty)\to\mathbb{R}^{m}$ is defined as
\begin{equation}
    e_{3,i}(t)\triangleq \hat{y}_i(t)-y_i(t).
\label{e3i}
\end{equation}
Contrary to \eqref{e1i}, the sampled state estimation error $e_{2,i}(t)$ and the output estimation error $e_{3,i}(t)$ are both measurable by agent $i$, where $e_{3,i}(t)$ is used to drive the estimate of agent $i$ towards the state of the system. Given the state estimation error in \eqref{e1i}, the sensor network is said to have successfully reconstructed the state of the system whenever
\begin{equation}
    \underset{t\to\infty}{\text{lim sup}} \ \Vert e_{1,i}(t) \Vert \leq \varepsilon \quad \forall i\in\mathcal{V}
\end{equation}
for some user-defined $\varepsilon>0$. 

\section{Observer Development} \label{Observer Section}
Let $\{t_{k}^{i}\}_{k=0}^{\infty}$ be an increasing sequence of event-times for agent $i$, where $t_{k}^{i}$ denotes the $k^{\text{th}}$ instance agent $i$ samples and broadcasts its state estimate of the system in \eqref{plant dynamics} to all agents $j\in\mathcal{N}_{i}$. Note that the broadcast information is received by all neighbors simultaneously, i.e., we assume perfect communication. The sampled state estimate of the system as computed by agent $i$ is defined as\footnote{A zero-order hold, i.e., sampled state estimate, is used in this work since the system dynamics are unknown. Future works can consider observers that allow $\tilde{x}_{i}(t)$ to vary over $t\in[t_{k}^{i},t_{k+1}^{i})$.}
\begin{equation}
    \tilde{x}_{i}(t)\triangleq\hat{x}_{i}(t_{k}^{i}),\text{ }t\in[t_{k}^{i},t_{k+1}^{i})
\label{x_i_tilde}
\end{equation} 
for all $j\in\mathcal{N}_{i}\cup\{i\}$, \update{where the state estimate $\hat{x}_{i}(t)$ is generated by the subsequently defined observer.} Hence, all neighbors of agent $i$, including agent $i$, have access to the synchronized sampled state estimate from agent $i$.
\begin{assumption}
    The state of the system in \eqref{plant dynamics} evolves within a compact set $\mathcal{D}\subset\mathbb{R}^{n}$ for all time, i.e., $x_{0}(t)\in\mathcal{D}$ for all $t\geq 0$.
\label{Assumption 4}
\end{assumption}

\noindent Since the nonlinear function $f$ is continuous and $x_{0}(t)$ is contained within a compact set by Assumption \ref{Assumption 4}, we can invoke the Stone-Weierstrass Theorem to express the nonlinear dynamics in \eqref{plant dynamics} within $\mathcal{D}$ as
\begin{equation}
    f(x_{0}(t))=W_{0}^{\top}\sigma(\Phi(x_{0}(t)))+\varepsilon(x_{0}(t)),
\label{NN for drift dynamics}
\end{equation}
where $W_{0}\in\mathbb{R}^{L\times n}$ denotes the ideal outer layer weight matrix, $\sigma:\mathbb{R}^{p}\to\mathbb{R}^{L}$ denotes a vector containing bounded continuous activation functions\footnote{Examples of continuous activation functions are the sigmoid function, the hyperbolic tangent, and the Gaussian function.}, $\Phi:\mathbb{R}^{n}\to\mathbb{R}^{p}$ encodes the ideal inner DNN, and $\varepsilon:\mathbb{R}^{n}\to\mathbb{R}^{n}$ denotes the bounded function reconstruction error \update{\cite[Theorem 7.32]{rudin1976principles}}. Note that $W_{0}$, $\Phi$, and $\varepsilon$ are unknown. The ideal inner DNN can be expressed as
\begin{equation}
    \Phi(x_{0}(t))=(W_{\ell}^{\top}\phi_{\ell}\circ W_{\ell-1}^{\top}\phi_{\ell-1}\circ...\circ W_{1}^{\top}\phi_{1})(x_{0}(t)),
\label{Phi func}
\end{equation}
where $\text{\ensuremath{\ell}}\in\mathbb{Z}_{\geq 1}$ denotes the number of user-defined inner layers of the DNN, $q\in[\ell]$, $W_{q}\in\mathbb{R}^{L_{q}\times n_{q+1}}$ denotes the ideal weight matrix for the $q^{\text{th}}$ inner layer, and $\phi_{q}:\mathbb{R}^{n_{q}}\to\mathbb{R}^{L_{q}}$ denotes a vector function composed of scalar basis functions corresponding to the $q^{\text{th}}$ inner layer. Note that $n_{1}=n$ and $n_{\ell+1}=p$. Moreover, for each $q\in[\text{\ensuremath{\ell}}]$, $W_{q}$ is unknown. Using \eqref{NN for drift dynamics}, the system model in \eqref{plant dynamics} can be expressed as 
\begin{equation}
    \dot{x}_{0}(t) = W_{0}^{\top}\sigma(\Phi(x_{0}(t)))+\varepsilon(x_{0}(t))+d(t).
\label{Alternative Plant Dynamics}
\end{equation}
\update{Based on \eqref{Alternative Plant Dynamics} and the subsequent stability analysis}, the distributed observer of agent $i\in\mathcal{V}$ is defined as
\begin{equation}
\begin{aligned}
    \dot{\hat{x}}_{i}(t) \triangleq & \ 	\widehat{W}_{i}^{\top}(t)\sigma(\widehat{\Phi}_{i}(\hat{x}_{i}(t)))+K_{1}(z_{i}(t)-C_{i}^{\top}e_{3,i}(t)),\\
z_{i}(t)	\triangleq & \	\underset{j\in\mathcal{N}_{i}}{\sum}a_{ij}\left(\tilde{x}_{j}\left(t\right)-\tilde{x}_{i}\left(t\right)\right),\\
\hat{y}_{i}(t)	\triangleq & \ C_{i}\hat{x}_{i}(t),
\end{aligned}
\label{Observer i}
\end{equation}
where $\widehat{W}_{i}:[0,\infty)\to\mathbb{R}^{L\times n}$ denotes the estimated outer weight matrix of the system as computed by agent $i$, $\widehat{\Phi}_{i}:\mathbb{R}^{n}\to\mathbb{R}^{p}$ encodes the estimated inner DNN computed by agent $i$, and $K_{1}\in\mathbb{R}^{n\times n}$ is the \alert{symmetric} solution to the bilinear matrix inequality 
\begin{equation}
    \frac{1}{2}(I_{N}\otimes K_{1})C^{\top}C+\frac{1}{2}C^{\top}C(I_{N}\otimes K_{1})+(L\otimes K_{1})\geq k_{1}I_{nN}.
\label{BMI}
\end{equation}
Observe that $C \triangleq \text{diag}(C_{1},C_{2},...,C_{N})\in\mathbb{R}^{mN\times nN}$ denotes the output matrix of the sensor network, and $k_{1}\in\mathbb{R}_{> 0}$ is a user-defined parameter. The bilinear matrix inequality in \eqref{BMI} encodes an observability condition that originates from the subsequent stability analysis (see Section \ref{Stability Analysis}). The estimated inner DNN $\widehat{\Phi}_{i}(\hat{x}_{i}(t))$ is modeled as a piecewise continuous function \update{that is similar to \eqref{Phi func}}, where $T_{p}^{i}$ denotes the $p^{\text{th}}$ instance agent $i$ updates its DNN by training on collected input-output data. Hence, the set of discontinuities of $\widehat{\Phi}_{i}(\hat{x}_{i}(t))$ is given by $\{T_{p}^{i}\}_{p=1}^{\infty}$. 

\begin{remark}
Agent $i$ may collect input-output data from the system in \eqref{plant dynamics} and train a new inner DNN while the weights and biases of the previous inner DNN are held constant. Once training is complete, the new inner DNN can be switched in via \eqref{Observer i} \cite{Sun.Greene.eatoappear}.
\end{remark}

\noindent The error between the ideal outer weight matrix and the estimated outer weight matrix of agent $i$, i.e., $\widetilde{W}_{i}:[0,\infty)\to\mathbb{R}^{L \times n}$, is defined as
\begin{equation}
    \widetilde{W}_{i}(t) \triangleq W_{0}-\widehat{W}_{i}(t).
\label{Outer Weight Error}
\end{equation}
Since the ideal outer weights are unknown, $\widetilde{W}_{i}(t)$ is not measurable. Based on the subsequent Lyapunov stability analysis, the outer weight update law of agent $i$, which is embedded within the continuous projection operator denoted by $\text{proj}(\cdot,\cdot)$ and defined in \cite[Equation 4]{cai2006sufficiently}, is designed as\footnote{\alert{The projection operator is used to ensure $\widehat{W}_{i}(t)$ remains within the set $\Omega\triangleq\{w\in\mathbb{R}^{L\times n}:\Vert w \Vert \leq \bar{\omega}\}$ for all $t\geq 0$, where $\bar{\omega}\in\mathbb{R}_{>0}$ is a user-defined parameter.}} \alert{
\begin{equation}
\begin{aligned}
    \dot{\omega}_{i}(t) = & \text{ proj}(\mu_i,\omega_i),\\
    \mu_i \triangleq & -\text{vec}(\Gamma_{i}\sigma(\widehat{\Phi}_{i}(\hat{x}_{i}(t)))e_{3,i}^{\top}(t)C_{i}),\\
    \omega_i \triangleq & \text{ vec}(\widehat{W}_i(t)),
\end{aligned}
\label{What_i Dot}
\end{equation}
}where $\Gamma_{i}\in\mathbb{R}^{L\times L}$ is a user-defined positive definite matrix used to adjust the learning rate of the outer layer weights for the DNN of agent $i$. \update{Moreover, since the inner DNN is treated as an arbitrary piecewise continuous function, the user is free to employ virtually any offline training policy. An example policy is discussed in Section \ref{Simulation Section} and additional training strategies can be found in \cite{joshi2019deep} and \cite{Joshi.Virdi.ea.2020}.}

Using \eqref{y_i}, \eqref{e1i}, and the definition of $\hat{y}_{i}(t)$ in \eqref{Observer i}, $e_{3,i}(t)$ can be alternatively expressed as
\begin{equation}
    e_{3,i}(t) = C_{i}e_{1,i}(t).
\label{Alternative e3i}
\end{equation}
Substituting \eqref{e1i} and \eqref{e2i} into the definition of $z_{i}(t)$ in \eqref{Observer i} yields
\begin{equation}
    z_{i}(t) = \underset{j\in\mathcal{N}_{i}}{\sum}a_{ij}(e_{1,j}(t)-e_{1,i}(t))+\underset{j\in\mathcal{N}_{i}}{\sum}a_{ij}(e_{2,j}(t)-e_{2,i}(t)).
\label{Alternative z_i}
\end{equation}

\noindent The expression in \eqref{Alternative z_i} is not measurable since it contains state estimation errors for agents $j\in\mathcal{N}_{i}\cup\{i\}$. However, the expression for $z_{i}(t)$ in \eqref{Observer i} is measurable and equivalent to \eqref{Alternative z_i}, where \eqref{Alternative z_i} is used in the analysis. The closed-loop error dynamics of $e_{1,i}(t)$ can now be determined by substituting \eqref{Alternative Plant Dynamics}--\eqref{Outer Weight Error}, \eqref{Alternative e3i}, and \eqref{Alternative z_i} into the time derivative of \eqref{e1i}, when it exists, while adding and subtracting $W_{0}^{\top}\sigma(\widehat{\Phi}_{i}(\hat{x}_{i}(t)))$ to obtain
\begin{equation}
\begin{aligned}
    \dot{e}_{1,i}(t) & = -\widetilde{W}_{i}^{\top}(t)\sigma(\widehat{\Phi}_{i}(\hat{x}_{i}(t))) -K_{1}C_{i}^{\top}C_{i}e_{1,i}(t)+\chi_{i}(t) \\
    & +K_{1}\underset{j\in\mathcal{N}_{i}}{\sum}a_{ij}(e_{1,j}(t)-e_{1,i}(t))\\
    & +K_{1}\underset{j\in\mathcal{N}_{i}}{\sum}a_{ij}(e_{2,j}(t)-e_{2,i}(t)),
\end{aligned}
\label{e1i Dot}
\end{equation}
where $\chi_{i}(t)\triangleq W_{0}^{\top}(\sigma(\widehat{\Phi}_{i}(\hat{x}_{i}(t)))-\sigma(\Phi(x_{0}(t))))-\varepsilon(x_{0}(t))-d(t)\in\mathbb{R}^{n}$. 

To express the subsequent development in a compact form, let $e_{1}(t) \triangleq [e_{1,1}^{\top}(t),e_{1,2}^{\top}(t),...,e_{1,N}^{\top}(t)]^{\top}\in\mathbb{R}^{nN}$, $e_{2}(t) \triangleq$ \\ $[e_{2,1}^{\top}(t),e_{2,2}^{\top}(t),...,e_{2,N}^{\top}(t)]^{\top}\in\mathbb{R}^{nN}$, $e_{3}(t)\triangleq [e_{3,1}^{\top}(t),e_{3,2}^{\top}(t)$ \\ $,...,e_{3,N}^{\top}(t)]^{\top}\in\mathbb{R}^{mN}$, $\hat{x}(t) \triangleq [\hat{x}_{1}^{\top}(t),\hat{x}_{2}^{\top}(t),...,\hat{x}_{N}^{\top}(t)]^{\top}\in\mathbb{R}^{nN}$, $\chi(t)  \triangleq [\chi_{1}^{\top}(t),\chi_{2}^{\top}(t),...,\chi_{N}^{\top}(t)]^{\top}\in\mathbb{R}^{nN}$, and $z(t) \triangleq [z_{1}^{\top}(t),z_{2}^{\top}(t),...,z_{N}^{\top}(t)]^{\top}\in\mathbb{R}^{nN}$. The block diagonal matrix composed from the DNN outer weight errors is $\widetilde{W}(t) \triangleq \text{diag}(\widetilde{W}_{1}(t),\widetilde{W}_{2}(t),...,\widetilde{W}_{N}(t))\in\mathbb{R}^{LN\times nN}$. Similarly, the block diagonal matrix consisting of the DNN outer weight estimates is $\widehat{W}(t) \triangleq \text{diag}(\widehat{W}_{1}(t),\widehat{W}_{2}(t),...,\widehat{W}_{N}(t))\in\mathbb{R}^{LN\times nN}$. Let $\sigma_{i}(t) \triangleq \sigma(\widehat{\Phi}_{i}(\hat{x}_{i}(t)))$ for every agent $i\in\mathcal{V}$. The block diagonal matrices consisting of $\{\Gamma_{i}\}_{i\in\mathcal{V}}$ and the DNN components of all agents prior to being multiplied by their corresponding outer weights are denoted by $\Gamma \triangleq \text{diag}(\Gamma_{1},\Gamma_{2},...,\Gamma_{N})\in\mathbb{R}^{LN\times LN}$ and $\sigma(\widehat{\Phi}(\hat{x}(t))) \triangleq[\sigma_{1}^{\top}(t),\sigma_{2}^{\top}(t),...,\sigma_{N}^{\top}(t)]^{\top}\in\mathbb{R}^{LN}$, respectively. 

Using \eqref{e1i Dot} and the stacked expressions for $e_{1}(t)$, $e_{2}(t)$, $\chi(t)$, $C$, $\widetilde{W}(t)$, $\widehat{W}(t)$, and $\sigma(\widehat{\Phi}(\hat{x}(t)))$, the closed-loop dynamics of $e_{1}(t)$ are
\begin{equation}
\begin{aligned}
    \dot{e}_{1}(t) & = -\widetilde{W}^{\top}(t)\sigma(\widehat{\Phi}(\hat{x}(t))) -(I_{N}\otimes K_{1})C^{\top}Ce_{1}(t)\\
    & -(L\otimes K_{1})e_{1}(t) -(L\otimes K_{1})e_{2}(t) + \chi(t).
\end{aligned} 
\label{e1 Dot}
\end{equation}
Since $W_{0}$ is a fixed matrix, $\sigma(\cdot)$ is a bounded function, the function reconstruction error $\varepsilon$ is bounded, and the disturbance is bounded given Assumption \ref{Assumption  2}, there exists a constant $\chi_{\max}\in\mathbb{R}_{>0}$ such that $\Vert \chi(t) \Vert \leq \chi_{\max}$ for all $t\geq 0$.  
Moreover, from \eqref{Alternative z_i}, it follows that $z(t)$ can be expressed as
\begin{equation}
    z\left(t\right) = -(L\otimes I_{n})e_{1}(t)-(L\otimes I_{n})e_{2}(t).
\label{z Ensemble}
\end{equation}
Using \eqref{z Ensemble} and Young's inequality, it follows that
\begin{equation}
    -\Vert e_{1}(t)\Vert ^{2} \leq \Vert e_{2}(t)\Vert ^{2} - \frac{1}{2 \Vert L\otimes I_{n} \Vert ^{2}}\Vert z(t)\Vert ^{2},
\label{Trigger Inequality}
\end{equation}
which is a useful inequality employed in the development of the event-trigger mechanism for the sensor network. 

\section{Stability Analysis} \label{Stability Analysis}
The following objects are presented to facilitate the development. Let $\bar{\chi}(t) \triangleq \chi(t)-(I_{nN}-C^{\top}C)\widetilde{W}^{\top}(t)\sigma(\widehat{\Phi}(\hat{x}(t)))$. Observe that $C$ is bounded by construction. Similarly, $\widetilde{W}(t)$ is bounded since $W_{0}$ is fixed and the projection operator ensures $\widehat{W}_{i}(t)$ is bounded for each $i\in\mathcal{V}$. Moreover, $\sigma(\cdot)$ is bounded by design. Hence, there exists a $\bar{\chi}_{\max}\in\mathbb{R}_{>0}$ such that $\Vert \bar{\chi}(t) \Vert \leq \bar{\chi}_{\max}$ for all $t\geq 0$. Furthermore, there exists a constant $\widetilde{W}_{\max}\in\mathbb{R}_{>0}$ such that $\frac{1}{2}\text{tr}\left(\widetilde{W}^{\top}(t)\Gamma^{-1}\widetilde{W}(t)\right)\leq\widetilde{W}_{\max}$, which can be made arbitrarily small through the choice of $\Gamma$. Let $k_{1}\triangleq k_{2}+\frac{\rho^{2}}{\delta}$, $\alpha \triangleq k_{2}-\frac{1}{\kappa}$, and $\bar{\delta}\triangleq\delta+\epsilon$. Select $\kappa>0$, $k_{2}>\frac{1}{\kappa}$, $\delta > 0$, $\rho\geq\bar{\chi}_{\max}$, and $\epsilon > 0$. Hence, $k_{1}>0$, $\bar{\delta}>0$, and $\alpha>0$.

The event-times $\{t_{k}^{i}\}_{k=0}^{\infty}$ that dictate when agent $i$ samples and broadcasts its state estimate $\hat{x}_{i}(t)$, as outlined in \eqref{x_i_tilde}, are generated by the event-trigger mechanism
\begin{equation}
\begin{aligned}
t_{k+1}^{i}\triangleq & \inf\left\{ t>t_{k}^{i}: \phi_{1}\Vert e_{2,i}(t)\Vert ^{2}\geq \phi_{2}\Vert z_{i}(t)\Vert ^{2}+\frac{\epsilon}{N}\right\},\\
\phi_{1} \triangleq & \frac{k_{2}}{2}+\frac{\kappa}{2}\Vert L\otimes K_{1}\Vert ^{2}, \quad \phi_{2} \triangleq \frac{k_{2}}{4\Vert L\otimes I_{n}\Vert ^{2}}.
\end{aligned}
\label{Trigger}
\end{equation}
\noindent Since $k_{2}$ and $\kappa$ are positive, $\phi_{1}>0$ and $\phi_{2}>0$.\footnote{The piecewise continuity of $e_{2,i}(t)$, $\epsilon>0$, and \eqref{Trigger} can be used to show that, after each event-time of agent $i$, there exists a well-defined time interval over which agent $i$ does not trigger.} Moreover, the event-trigger mechanism in \eqref{Trigger} originates from the subsequent stability analysis. \update{Notice that \eqref{Trigger} requires each agent to sample their state estimate whenever the error between the sampled and continuous estimates becomes sufficiently large.}

\begin{theorem} \label{Theorem 1}
The observer in \eqref{Observer i} and update law in \eqref{What_i Dot} for each $i\in\mathcal{V}$ ensure the state estimation error $e_{1}(t)$ is UUB in the sense that 
\begin{equation}
\begin{aligned}
    \Vert e_{1}(t)\Vert^{2} &\leq \left(\Vert e_{1}(0)\Vert ^{2}+\text{tr}(\widetilde{W}^{\top}(0)\Gamma^{-1}\widetilde{W}(0))\right)e^{-\alpha t}\\
    & +2\left(\widetilde{W}_{\max}+\frac{\bar{\delta}}{\alpha}\right)(1-e^{-\alpha t})
\end{aligned}
\label{UUB Bound}
\end{equation}
provided Assumptions \ref{Assumption 1}--\ref{Assumption 4} are satisfied, there exists a matrix $K_{1}$ satisfying the bilinear matrix inequality in \eqref{BMI}, and agent $i$ broadcasts its state estimate as determined by the event-trigger mechanism in \eqref{Trigger} for each $i\in\mathcal{V}$.
\end{theorem}

\begin{proof}
Consider the Lyapunov function candidate $V:\mathcal{D}\to\mathbb{R}_{\geq 0}$ defined as
\begin{equation}
    V(\zeta(t)) \triangleq \frac{1}{2}e_{1}^{\top}(t)e_{1}(t) + \frac{1}{2}\text{tr}(\widetilde{W}^{\top}(t)\Gamma^{-1}\widetilde{W}(t)),
\label{V Lyap}
\end{equation}
where $\zeta(t) \triangleq [e_{1}^{\top}(t),\text{vec}(\widetilde{W}(t))^{\top}]^{\top}\in\mathbb{R}^{nN+nLN^2}$. Observe that $V(\zeta(t))$ can be bounded as
\begin{equation}
    \frac{1}{2}e_{1}^{\top}(t)e_{1}(t) \leq V(\zeta(t)) \leq \frac{1}{2}e_{1}^{\top}(t)e_{1}(t)+\widetilde{W}_{\max}.
\label{V Bound}
\end{equation}

\noindent Suppose $\xi:\left[0,\infty\right)\rightarrow\mathbb{R}^{nN+nLN^2}$ is a Filippov solution to the differential inclusion $\dot{\xi}(t)\in K[\mathcal{H}](\xi(t))$, where $\xi(t)=\zeta(t)$, the mapping $K[\cdot]$ provides a calculus for computing Filippov's differential inclusion as defined in \cite{paden1987calculus}, and $\mathcal{H}:\mathbb{R}^{nN+nLN^2}\rightarrow\mathbb{R}^{nN+nLN^2}$ is defined as $\mathcal{H}(\xi(t))=[\dot{e}_{1}^{\top}(t),\text{vec}(\dot{\widetilde{W}}(t))^{\top}]^{\top}.$ The time derivative of $V(\zeta(t))$ exists almost everywhere (a.e.) and
\begin{equation}
    \dot{V}(\xi(t)) \overset{a.e.}{\in}\dot{\widetilde{V}}(\xi(t)),
\label{AE Condition}
\end{equation}
where $\dot{\widetilde{V}}(\xi(t))$ is the generalized time derivative of $V(\zeta(t))$ along the Filippov trajectories of $\dot{\xi}(t)=\mathcal{H}(\xi(t))$. By \cite[Equation 13]{shevitz1994lyapunov},
\begin{equation}
    \dot{\widetilde{V}}(\xi(t)) \triangleq \underset{\eta\in\partial V(\xi(t))}{\bigcap}\eta^{\top}\left[K[\mathcal{H}]^{\top}(\xi(t)), \ 1\right]^{\top},
\end{equation}
where $\partial V(\xi(t))$ denotes the Clarke generalized gradient of $V(\xi(t))$. Since $V(\xi(t))$ is continuously differentiable in $\xi(t)$, $\partial V(\xi(t))=\{ \nabla V(\xi(t))\}$, where $\nabla$ denotes the gradient operator. Using the calculus of $K[\cdot]$ from \cite{paden1987calculus} and simplifying the substitution of \eqref{e1 Dot} into the generalized time derivative of \eqref{V Lyap}, one has
\begin{equation}
\begin{aligned}
    \dot{\widetilde{V}}(\xi(t)) \subseteq & -e_{1}^{\top}(t)\widetilde{W}^{\top}(t)K\left[\sigma(\widehat{\Phi}(\hat{x}(t)))\right]+ e_{1}^{\top}(t)K[\chi(t)]\\
    - & \left\{e_{1}^{\top}(t)(L\otimes K_{1})e_{1}(t)\right\}-e_{1}^{\top}(t)(L\otimes K_{1})K[e_{2}(t)]\\
    - & \left\{e_{1}^{\top}(t)(I_{N}\otimes K_{1})C^{\top}Ce_{1}(t)\right\}\\
    - & \text{tr}\left(\widetilde{W}^{\top}(t)\Gamma^{-1}K\left[\dot{\widehat{W}}(t))\right]\right). 
\end{aligned}
\label{V Generalized Dot}
\end{equation}
Using the estimated outer weight update law in \eqref{What_i Dot} for each $i\in\mathcal{V}$ and the stacked expressions for $\widetilde{W}(t)$, $\widehat{W}(t)$, $e_{3}(t)$, $C$, $\Gamma$, and $\sigma(\widehat{\Phi}(\hat{x}(t)))$, the time derivative of $\frac{1}{2}\text{tr}(\widetilde{W}^{\top}(t)\Gamma^{-1}\widetilde{W}(t))$ yields
\alert{
\begin{equation}
\begin{aligned}
    \text{tr}(\widetilde{W}^{\top}(t)\Gamma^{-1}\dot{\widehat{W}}(t))  
    & = \sum_{i\in\mathcal{V}}\text{tr}(\widetilde{W}_i^{\top}(t)\Gamma_i^{-1}\dot{\widehat{W}}_i(t))\\
    & = \sum_{i\in\mathcal{V}} \text{vec}(\Gamma_i^{-1}\widetilde{W}_i(t))^{\top}\text{proj}(\mu_i,\omega_i)\\
    & \geq -e_{3}^{\top}(t)C\widetilde{W}^{\top}(t)\sigma(\widehat{\Phi}(\hat{x}(t))).
\end{aligned}
\label{Trace Term}
\end{equation}
}Substituting \eqref{Alternative e3i} for all $i\in\mathcal{V}$ into $e_{3}(t)$ yields $e_{3}(t) = Ce_{1}(t)$. Adding and subtracting $C^{\top}C$ while using $e_{3}(t) = Ce_{1}(t)$ results in
\begin{equation}
\begin{aligned}
    e_{1}^{\top}(t)\widetilde{W}^{\top}(t)\sigma(\widehat{\Phi}(\hat{x}(t))) = & \   e_{3}^{\top}(t)C\widetilde{W}^{\top}(t)\sigma(\widehat{\Phi}(\hat{x}(t)))\\
    & + \ e_{1}^{\top}(t)(I_{nN}-C^{\top}C)\widetilde{W}^{\top}(t)\\
    & \cdot \sigma(\widehat{\Phi}(\hat{x}(t))).
\end{aligned}
\label{Auxiliary Term}
\end{equation}
Substituting \eqref{Trace Term} and \eqref{Auxiliary Term} into \eqref{V Generalized Dot} and utilizing \eqref{AE Condition}, it follows that
\begin{equation}
\begin{aligned}
    \dot{V}(\zeta(t)) \update{\overset{a.e.}{\leq}} & -e_{3}^{\top}(t)C\widetilde{W}^{\top}(t)\sigma(\widehat{\Phi}(\hat{x}(t)))\\
    & - e_{1}^{\top}(t)(I_{nN}-C^{\top}C)\widetilde{W}^{\top}(t)\sigma(\widehat{\Phi}(\hat{x}(t)))\\
    & -e_{1}^{\top}(t)(L\otimes K_{1})e_{1}(t)-e_{1}^{\top}(t)(L\otimes K_{1})e_{2}(t)\\
    & -e_{1}^{\top}(t)(I_{N}\otimes K_{1})C^{\top}Ce_{1}(t)+e_{1}^{\top}(t)\chi(t)\\
    & +e_{3}^{\top}(t)C\widetilde{W}^{\top}(t)\sigma(\widehat{\Phi}(\hat{x}(t)))\\
    \update{\overset{a.e.}{\leq}} & -e_{1}^{\top}(t)(L\otimes K_{1})e_{1}(t)-e_{1}^{\top}(t)(L\otimes K_{1})e_{2}(t)\\
    & -e_{1}^{\top}(t)(I_{N}\otimes K_{1})C^{\top}Ce_{1}(t)+e_{1}^{\top}(t)\tilde{\chi}(t).
\end{aligned}
\label{V Dot One}
\end{equation}
Using Young's inequality, \eqref{V Dot One} can be upper bounded as
\begin{equation}
\begin{aligned}
    \dot{V}(\zeta(t)) \overset{a.e.}{\leq} & -e_{1}^{\top}(t)\bigg(\frac{1}{2}(I_{N}\otimes K_{1})C^{\top}C+\frac{1}{2}C^{\top}C(I_{N}\otimes K_{1})\\
    & +(L\otimes K_{1})\bigg)e_{1}(t) +\frac{\kappa}{2}\Vert L\otimes K_{1}\Vert ^{2}\Vert e_{2}(t)\Vert ^{2}\\
    & +\frac{1}{2\kappa}\Vert e_{1}(t)\Vert ^{2} +\bar{\chi}_{\max}\Vert e_{1}(t)\Vert. 
\end{aligned}
\label{V Dot Two}
\end{equation}
Using the bilinear matrix inequality in \eqref{BMI} and $k_{1} = k_{2}+\frac{\rho^{2}}{\delta}$, \eqref{V Dot Two} can be upper bounded by
\begin{equation}
\begin{aligned}    
    \dot{V}(\zeta(t)) \overset{a.e.}{\leq} & -\left(k_{2}-\frac{1}{2\kappa}\right)\Vert e_{1}(t)\Vert ^{2}+\frac{\kappa}{2}\Vert L\otimes K_{1}\Vert ^{2}\Vert e_{2}(t)\Vert ^{2}\\
    & +\bar{\chi}_{\max}\Vert e_{1}(t)\Vert -\frac{\rho^{2}}{\delta}\Vert e_{1}(t)\Vert ^{2}.
\end{aligned} 
\label{V Dot Three}
\end{equation}
Since $\rho\geq\bar{\chi}_{\max}$, it follows that $\bar{\chi}_{\max}\Vert e_{1}(t)\Vert -\frac{\rho^{2}}{\delta}\Vert e_{1}(t)\Vert ^{2}\leq\delta$. Using this inequality and \eqref{Trigger Inequality}, \eqref{V Dot Three} can be upper bounded as
\begin{equation}
\begin{aligned}
    \dot{V}(\zeta(t)) \overset{a.e.}{\leq} &
    -\frac{1}{2}\left(k_{2}-\frac{1}{\kappa}\right)\Vert e_{1}(t)\Vert ^{2} + \bar{\delta}\\
    & + \underset{i\in\mathcal{V}}{\sum} \left(\phi_{1}\Vert e_{2,i}(t)\Vert ^{2} - \phi_{2}\Vert z_{i}(t)\Vert ^{2}-\frac{\epsilon}{N}\right).
\end{aligned}
\label{V Dot Four}
\end{equation}
Based on \eqref{V Dot Four}, the event-trigger mechanism for each agent $i\in\mathcal{V}$ is given by \eqref{Trigger}. Since each agent provides state feedback according to the event-trigger mechanism in \eqref{Trigger}, \eqref{V Dot Four} can be upper bounded as
\begin{equation}
    \dot{V}(\zeta(t)) \overset{a.e.}{\leq} -\frac{1}{2}\alpha\Vert e_{1}(t)\Vert ^{2} + \bar{\delta},
\label{V Dot Five}
\end{equation}
where $\alpha = k_{2}-\frac{1}{\kappa}$. Using \eqref{V Bound}, \eqref{V Dot Five} can be upper bounded as 
\begin{equation}
    \dot{V}(\zeta(t)) \overset{a.e.}{\leq} -\alpha V(\zeta(t)) + \alpha\widetilde{W}_{\max}+\bar{\delta}.
\label{V Dot Six}
\end{equation}
Note that $V(\zeta(t))$ is continuous over $\mathbb{R}_{>0}$, and $\dot{V}(\zeta(t))$ is continuous almost everywhere in $\mathbb{R}_{>0}$. The discontinuities of $\dot{V}(\zeta(t))$ occur over the set $\bigcup_{k,p\in\mathbb{Z}_{\geq 0}} \bigcup_{i\in\mathcal{V}} \{ t_{k}^{i}\}\cup\{T_{p}^{i}\}$, which is countable. Integrating \eqref{V Dot Six} yields
\begin{equation}
    V(\zeta(t)) \leq V(\zeta(0))e^{-\alpha t}+\left(\widetilde{W}_{\max}+\frac{\bar{\delta}}{\alpha}\right)(1-e^{-\alpha t}).
\label{V UUB Bound}
\end{equation}
Using \eqref{V Bound} and \eqref{V UUB Bound}, the result in \eqref{UUB Bound} follows.

We now show the constituent signals used in the observer are bounded. By Assumption \ref{Assumption 4}, $x_{0}(t)\in\mathcal{L}_{\infty}$. Since $e_{1}(t)\in\mathcal{L}_{\infty}$ given \eqref{UUB Bound}, the definition of $e_{1}(t)$ implies $e_{1,i}(t)\in\mathcal{L}_{\infty}$ for each $i\in\mathcal{V}$. Since $x_{0}(t)\in\mathcal{L}_{\infty}$ and $e_{1,i}(t)\in\mathcal{L}_{\infty}$ for each $i\in\mathcal{V}$, \eqref{e1i} implies $\hat{x}_{i}(t)\in\mathcal{L}_{\infty}$ for each $i\in\mathcal{V}$. Since $\hat{x}_{i}(t)\in\mathcal{L}_{\infty}$, \eqref{x_i_tilde} implies $\tilde{x}_{i}(t)\in\mathcal{L}_{\infty}$. Since $\tilde{x}_{i}(t)\in\mathcal{L}_{\infty}$ for each $i\in\mathcal{V}$, \eqref{Observer i} implies $z_{i}(t)\in\mathcal{L}_{\infty}$. Since $\hat{x}_{i}(t)\in\mathcal{L}_{\infty}$ and $C_{i}$ is a fixed matrix, \eqref{Observer i} implies $\hat{y}_{i}(t)\in\mathcal{L}_{\infty}$. Since $x_{0}(t)\in\mathcal{L}_{\infty}$ and $C_{i}$ is a fixed matrix, \eqref{y_i} implies $y_{i}(t)\in\mathcal{L}_{\infty}$. Since $\hat{y}_{i}(t)\in\mathcal{L}_{\infty}$ and $y_{i}(t)\in\mathcal{L}_{\infty}$, \eqref{e3i} implies $e_{3,i}(t)\in\mathcal{L}_{\infty}$. Lastly, $\widehat{W}_{i}(t)\in\mathcal{L}_{\infty}$ by construction.
\end{proof}

\begin{theorem}
The difference between consecutive broadcast times generated by the event-trigger mechanism of agent $i\in\mathcal{V}$ in \eqref{Trigger} is uniformly lower bounded by \begin{equation}
    t_{k+1}^{i}-t_{k}^{i}	\geq \frac{1}{e_{2,i}^{\max}}\sqrt{\frac{\epsilon}{N\phi_{1}}}
\label{Zeno Bound}
\end{equation} 
for all instances $k\in\mathbb{Z}_{\geq0}$, where $e_{2,i}^{\max}\in\mathbb{R}_{>0}$ is a bounding constant such that $\Vert \widehat{W}_{i}(t)\Vert \Vert \sigma(\widehat{\Phi}_{i}(\hat{x}_{i}(t)))\Vert +\Vert K_{1}\Vert \Vert z_{i}(t)\Vert +\Vert K_{1}C_{i}^{\top}\Vert \Vert e_{3,i}(t)\Vert \leq e_{2,i}^{\max}$ and $\phi_1$ is a constant defined in \eqref{Trigger}.
\end{theorem}

\begin{proof}
Let $t\geq t_{k}^{i}\geq0$ and $i\in\mathcal{V}$. Substituting \eqref{Observer i} into the time derivative of \eqref{e2i} yields
\begin{equation*}
    \dot{e}_{2,i}(t) \overset{a.e.}{=} -\widehat{W}_{i}^{\top}(t)\sigma(\widehat{\Phi}_{i}(\hat{x}_{i}(t)))-K_{1}\left(z_{i}(t)-C_{i}^{\top}e_{3,i}(t)\right),
\end{equation*}
where $\Vert \dot{e}_{2,i}(t)\Vert \overset{a.e.}{\leq} \Vert \widehat{W}_{i}(t)\Vert \Vert \sigma(\widehat{\Phi}_{i}(\hat{x}_{i}(t)))\Vert +\Vert K_{1}\Vert \Vert z_{i}(t)\Vert +\Vert K_{1}C_{i}^{\top}\Vert \Vert e_{3,i}(t)\Vert $. Recall that $z_{i}(t)\in\mathcal{L}_{\infty}$ and $e_{3,i}(t)\in\mathcal{L}_{\infty}$ from the proof of Theorem \ref{Theorem 1}. Moreover, $\widehat{W}_{i}(t)$ and $\sigma(\widehat{\Phi}_{i}(\hat{x}_{i}(t)))$ are bounded by construction. Therefore, there exists $e_{2,i}^{\max}\in\mathbb{R}_{>0}$ such that $\Vert \widehat{W}_{i}(t)\Vert \Vert \sigma(\widehat{\Phi}_{i}(\hat{x}_{i}(t)))\Vert +\Vert K_{1}\Vert \Vert z_{i}(t)\Vert +\Vert K_{1}C_{i}^{\top}\Vert \Vert e_{3,i}(t)\Vert \leq e_{2,i}^{\max}$. Observe that 
\begin{equation*}
    \frac{d}{dt}\Vert e_{2,i}(t)\Vert =\frac{e_{2,i}^{\top}(t)\dot{e}_{2,i}(t)}{\Vert e_{2,i}(t)\Vert }\overset{a.e.}{\leq}\Vert \dot{e}_{2,i}(t)\Vert.
\end{equation*}
Therefore, $\frac{d}{dt}\Vert e_{2,i}(t)\Vert \overset{a.e.}{\leq}e_{2,i}^{\max}$. Recalling that $\Vert e_{2,i}(t_{k}^{i})\Vert =0$ and integrating $\frac{d}{dt}\Vert e_{2,i}(t)\Vert \overset{a.e.}{\leq}e_{2,i}^{\max}$ over $[t_{k}^{i},\infty)$, it follows that $\Vert e_{2,i}(t)\Vert \leq e_{2,i}^{\max}(t-t_{k}^{i})$ for all $t\in[t_{k}^{i},\infty)$. Using $\Vert e_{2,i}(t)\Vert \leq e_{2,i}^{\max}(t-t_{k}^{i})$ and the triggering condition in \eqref{Trigger}, the inequality in \eqref{Zeno Bound} follows.
\end{proof}

\section{Simulation Results} \label{Simulation Section}
To examine the performance of the developed estimation strategy, two numerical simulations are provided for a Van der Pol oscillator, where the model in \eqref{plant dynamics} is used with
\begin{equation}
    f(x_{0}(t)) = 
    \begin{bmatrix}
        \mu (x(t)-x^{3}(t)/3-y(t))\\
        x(t)/\mu\\
        -\mu z(t)
    \end{bmatrix},
\label{VanderPol}
\end{equation}
$x_{0}(t)\triangleq [x(t),y(t),z(t)]^{\top}\in\mathbb{R}^{3}$, and $\mu=0.3$. The disturbance affecting the system is given by 
\begin{equation*}
    d(t) = \left[
    0.5 \text{sin}(3t),\ 0.75 \text{cos}(t),\ \text{cos}(3.75t)\right]^{\top}.
\end{equation*}
The sensor network used to reconstruct the state of the system consists of three agents, where the output matrix of Agents 1, 2 and 3 are $C_{1}=[1 \ 0 \ 0]$, $C_{2}=[0 \ 1 \ 0]$, and $C_{3}=[0 \ 0 \ 1]$, respectively. Hence, each agent can measure one component of the system's state, and all agents cover the entire state of the system. Preliminary observations show that the bilinear matrix inequality in \eqref{BMI} has a solution whenever the combined output measurements of the agents cover the entire system state and the graph $\mathcal{G}$ is connected. The adjacency matrix of the sensor network is 
\begin{equation*}
    \mathcal{A} =
    \begin{bmatrix}
    0 & 1 & 1\\
    1 & 0 & 0\\
    1 & 0 & 0
    \end{bmatrix},
\end{equation*}
where Agent 1 can communication with Agents 2 and 3, but Agents 2 and 3 can only communicate with Agent 1. The first simulation investigates the performance of the observer when the weights and biases of all inner DNNs are held constant and only the outer weights of the DNNs are updated to ensure stability. The second simulation is identical to the first with the exception that the weights and biases of all inner DNNS are updated using input-output data. The simulation parameters, used in both simulations, are $N = 3$, $\kappa = 1$, $\epsilon = 3\times 10^{3}$, $\Gamma = 3 I_{3}$, $\rho = 2$, $\delta = 0.3$, and $k_{2} = 10$. Therefore, $k_{1} = 23.3$, where the CVX MATLAB toolbox in \cite{cvx} and \cite{gb08} was used to compute a $K_{1}$ that satisfies \eqref{BMI} as \update{$K_1 = \text{diag}(134.86,263.23,263.23)$}. Both simulations were $40$ time units long and used an integration time-step of $10^{-3}$ time units. The initial condition of the system and Agents $1$--$3$ are $x_{0} = [5\ 7\ 8]^{\top}$ and $x_{i} = [0\ 0\ 0]^{\top}$ for $i\in[3]$, respectively. 

Each agent used a $5$-layer DNN to approximate the system dynamics. Layers $1$--$3$ each consist of $12$ nodes, Layer $4$ consists of $5$ nodes, and Layer $5$ consists of $3$ nodes. Each layer is affine, i.e., each activation function is scaled by a weight and shifted by a bias, and all nodes use the tangent-sigmoid activation function. Layers $1$--$4$ used weights and biases, while Layer $5$ only used weights (biases set equal to zero). Given the model in \eqref{VanderPol}, the dimension of the input and output of each agent's DNN is $3$. As the outer weights of agent $i$'s DNN are updated online using \eqref{What_i Dot}, the user is free to specify an offline update procedure for the weights and biases of the inner DNN. This decoupling in training structure leads to a multiple timescale adaptation technique. Additional information about such methods can also be found in \cite{joshi2020design} and \cite{Sun.Greene.eatoappear}.

Input-output data and the Levenberg-Marquardt algorithm are used by the Deep Learning MATLAB toolbox in \cite{MATLAB_DLTB} to train the weights and biases of the inner DNN for each agent. Ideally, each agent would utilize $x_{0}(t)$ as input data and $\dot{x}_{0}(t)-d(t)$ as output data to train the inner DNN approximation of $f$. However, each agent does not have access to $x_{0}(t)$, $\dot{x}_{0}(t)$, and $d(t)$. Therefore, the input-output data used to train Layers $1$--$4$ of the DNN for agent $i$ are $\hat{x}_{i}(t)$ and $\dot{\hat{x}}_{i}(t)$. Given a continuously differentiable function, such as the DNN of agent $i$ with fixed weights and biases, and labeled input-output data, the Levenberg-Marquardt algorithm identifies a local minimizer for the nonlinear least squares problem formed from the given model and data. The algorithm is a combination of gradient descent and the Gauss-Newton method, which leads to efficient computation \cite{lourakis2005brief}. Because the Levenberg-Marquardt algorithm relies on labeled input-output data, i.e., each input corresponds to a known particular output, the algorithm can be considered a supervised learning strategy.

The learning rate, i.e., the step size used to train the inner DNN weights and biases at each iteration, was lower bounded by $10^{-3}$. The loss function used for training was the mean squared error (MSE) between the estimated output generated by a known input and the corresponding known output. Each training iteration was set to last until the MSE was less than $10^{-2}$ or a maximum of $25$ training epochs had elapsed. For each training iteration, $70\%$ of the data was used for training, $15\%$ was used for validation, and $15\%$ was used for testing. The weights and biases of Layers $1$--$4$ for all DNNs were initialized as $0.1$. The weights of Layer $5$ for the DNN of each agent were identically initialized as
\begin{equation*}
    \widehat{W}_{i}(0) = 
    \begin{bmatrix}
    2.51  &  4.59  &  3.40\\
   -2.44  &  0.47  & -2.45\\
    0.05  & -3.61  &  3.14\\
    1.99  & -3.50  & -2.56\\
    3.90  & -2.42  &  4.29
    \end{bmatrix}.
\end{equation*}

For the second simulation, data collection and training were performed by each agent during the first half of the simulation, and the resulting DNN model was implemented during the second half of the simulation. Specifically, each agent collected data during $t\in[10,16]$ and trained their inner DNN during $t\in[16,20]$. Each agent then implemented their updated inner DNN model for $t\geq 20$. While data collection can begin at any time, we observed improved state reconstruction after the transient response elapsed.

The results of the simulations are illustrated in Figures \ref{Figure 1}--\ref{Figure 3} and Table \ref{Table 1}. Note that although the agents performed data collection, training, and implementation in a synchronous manner for simplicity, asynchronous learning cycles can also be employed. Figure \ref{Figure 1} shows the norm of the state estimation error of all agents, where the weights and biases of Layers $1$--$4$ are held constant and only the weights of Layer $5$ are updated. Similarly, Figure \ref{Figure 2} shows the norm of the state estimation error of all agents, where the weights and biases of Layers $1$--$4$ and the weights of Layer $5$ are updated. The black dashed line, black solid line, and red dashed line in Figure \ref{Figure 2} denote the start of the data collection process, the end of the data collection process and beginning of the training process, and the end of the training process and DNN application time, respectively. Figure \ref{Figure 3} shows the event-times of each agent for $t\in[19.75,20.25]$ for the second simulation. The average difference between consecutive broadcast times for Agents $1$, $2$, and $3$ were $0.0047$, $0.0089$, and $0.0095$ time units, respectively, where no significant difference between broadcast times for the two simulations was observed.

\begin{figure}
\begin{center} 
\includegraphics[scale=1]{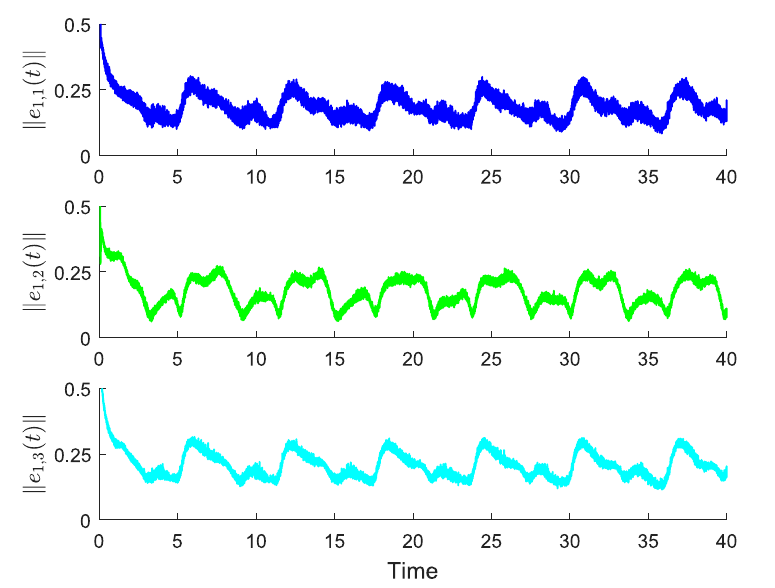} 
\caption {This figure illustrates the norm of the state estimation error, i.e., $\Vert e_{1,i}(t) \Vert$,  for each agent. The inner layer weights and biases of the DNNs are held constant for the entire simulation and only the outer layer weights are updated according to \eqref{What_i Dot}.}
\vspace{-10pt}
\label{Figure 1}
\end{center}
\end{figure}

\begin{figure}
\begin{center} 
\includegraphics[scale=1]{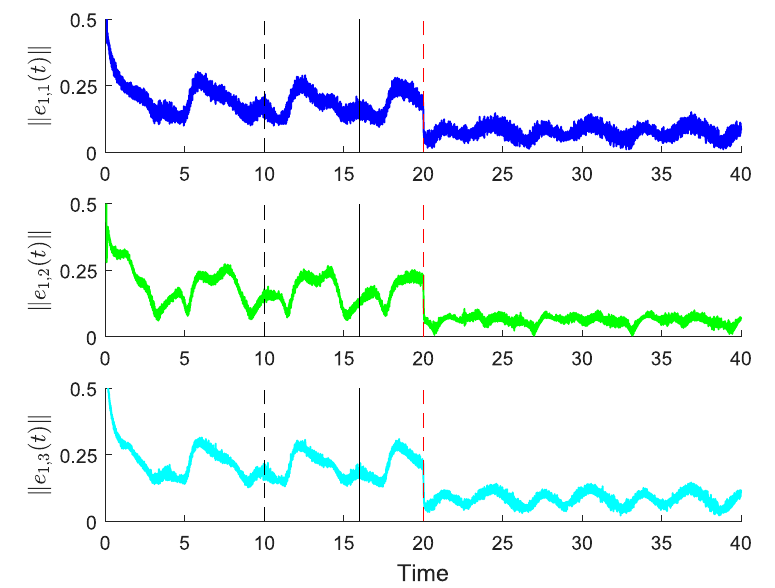} 
\caption {This figure portrays the norm of the state estimation error, i.e., $\Vert e_{1,i}(t) \Vert$, for each agent. When $t=10$, as denoted by the black dashed line, each agent begins collecting input-output data. Each agent stops their data collection process when $t=16$, which is denoted by the solid black line. For $t\in[16,20]$, each agent uses their collected data to train their inner DNN weights and biases, where training stops and the new DNN of each agent is switched in at $t=20$ according to \eqref{Observer i}. The red dashed line corresponds to the DNN implementation time of each agent.}
\vspace{-20pt}
\label{Figure 2}
\end{center}
\end{figure}

Table \ref{Table 1} lists the root-mean-square error (RMSE) of $\Vert e_{1,i}(t) \Vert$ for each agent during $t\geq 20$ for both simulations. Table \ref{Table 1} also lists the percent change between the two RMSE statistics, where the RMSE without DNN learning and with DNN learning define the initial and final values, respectively, used in the percent change computation. Figure \ref{Figure 2} and Table \ref{Table 1} demonstrate that the multiple timescale learning strategy can provide significant improvements in state reconstruction. A single training iteration of the inner DNNs reduced the RMSE of the state estimation error by approximately $60\%$ for each agent. While additional training iterations could have been performed, a single iteration is simulated since additional training cycles produced minimal improvements in state reconstruction for this case. 
\begin{figure}
\begin{center} 
\includegraphics[scale=1]{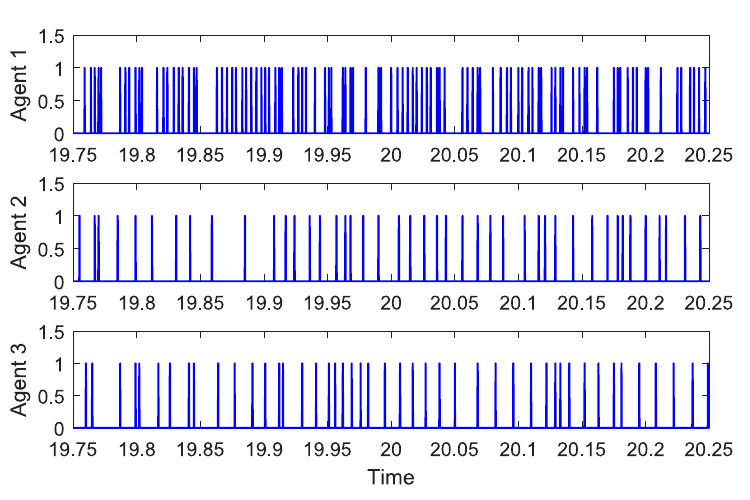} 
\caption {The event-times for each agent are depicted above for $t\in[19.75,20.25]$. A 0, or white space, denotes no communication, and a 1, or blue line, denotes a communication event. A $0.5$ time unit window of the simulation is shown, rather than entire simulation, to more clearly depict the intermittency and asynchrony in communication.}
\vspace{-20pt}
\label{Figure 3}
\end{center}
\end{figure}

\begin{table}[h!]
\centering
\caption{State Estimation RMSE With \& Without DNN Learning}
\vspace{-0.08 in}
\begin{tabular}{>{\centering\arraybackslash}p{0.25 in} >{\centering\arraybackslash}p{0.8 in} >{\centering\arraybackslash}p{0.8 in} >{\centering\arraybackslash}p{0.8 in}}
\hline
Agent   & RMSE Without DNN Learning   &  RMSE With DNN Learning   & Percent Change
\\
\hline 
   1   & 0.1827        & 0.0777          &   -57.47\%
\\ 
   2   & 0.1775        & 0.0634          &   -64.28\%        
\\ 
   3   & 0.2070        & 0.0876          &   -57.68\%       
\\
\hline
\end{tabular}
\\[0.0125 in]
\begin{flushleft} 
The root-mean-square of $\Vert e_{1,i}(t) \Vert$ for each agent is presented for both simulations, which was computed for $t\in[20,40]$.
\end{flushleft} 
\label{Table 1}
\end{table}

\section{Conclusion} \label{Concluding Section}
An adaptive event-triggered distributed state observer for a sensor network is developed, which is capable of reconstructing the state of an uncertain nonlinear system while being robust to a bounded disturbance. A DNN is used by each agent in the sensor network to approximate the uncertain nonlinear system dynamics from input-output data. The inner layer weights and biases are trained using the Levenberg-Marquardt algorithm in an offline manner, while the outer layer weights are updated using an analysis-based update law and real-time feedback. The result is a multiple timescale learning strategy. A nonsmooth Lyapunov stability analysis is provided that indicates the system's state can be uniformly reconstructed to within an ultimate bound. 

As seen in Figure \ref{Figure 1}, the observer drives the state reconstruction error for each agent below unity, which implies that the observer is capable of good performance without DNNs and that the observer parameters can be relaxed to decrease the frequency of communication. With respect to Figures \ref{Figure 1} and \ref{Figure 2}, it is evident that the use of DNNs allowed each agent to reduce their state reconstruction error by a substantial margin (approximately $60\%$ reduction). Moreover, the improvement in performance only required a single learning cycle, where additional learning cycles can be executed as necessary. A key observation noted during the simulation study is that training DNNs with data collected during the transient response leads to smaller performance improvements when compared to training using data collected during steady state. Furthermore, higher quality data, such as that generated from frequent communication, leads to improved performance. Future pursuits that could build on this result include using multiple DNNs to learn the system dynamics and disturbance separately, developing distributed consensus algorithms that share local weights and biases from the DNNs, and exploring the trade-off between triggering and state reconstruction using DNNs.

\bibliographystyle{IEEEtran}
\bibliography{Sources}

\end{document}